\documentclass[pra,
superscriptaddress,
nofootinbib,
twocolumn,
showpacs,preprintnumbers,amsmath,amssymb,a4wide]
{revtex4-1}

\usepackage{graphicx}
\usepackage{hyperref}
\usepackage{subfigure}
\usepackage{color}
\usepackage{braket}
\usepackage{bbm}
\usepackage{bm}
\usepackage{epsfig}
\usepackage{amsmath, amsthm, amsfonts}
\newtheorem{prop}{Proposition}
\newtheorem*{problem}{Problem}
\newtheorem*{modelloc}{$BGS$ model for the local part}
\newtheorem*{modellocgen}{Complete model for the local part}

\newcommand{\sgn}{\operatorname{sgn}}

\begin{document}

\title{The local content of all pure two-qubit states}

\author{Samuel Portmann}
\affiliation{Group of Applied Physics, University of Geneva, Chemin de Pinchat 22, CH-1211 Geneva 4, Switzerland}

\author{Cyril Branciard}
\affiliation{School of Mathematics and Physics, The University of Queensland, St Lucia, QLD 4072, Australia}

\author{Nicolas Gisin}
\affiliation{Group of Applied Physics, University of Geneva, Chemin de Pinchat 22, CH-1211 Geneva 4, Switzerland}


\date{\today}

\begin{abstract}
The \emph{(non-)local content} in the sense of  Elitzur, Popescu, and Rohrlich (EPR2) [Phys.~Lett.~A {\bf 162}, 25 (1992)] is a natural measure for the (non-)locality of quantum states. Its computation is in general difficult, even in low dimensions, and is one of the few open questions about pure two-qubit states. We present a complete solution to this long-lasting problem.
\end{abstract}


\maketitle
\thispagestyle{empty} 
John Bell showed in 1964~\cite{bell64} that the ``classical'' metaphysical assumptions used by Einstein, Podolsky and Rosen (EPR)~\cite{epr,footnote.einstein} have empirical consequences in the form of a constraint on correlations of quantum measurement outcomes: the Bell inequalities. Bell later subsumed these assumptions into the notion of \emph{local causality}~\cite{bell76}, today commonly abbreviated to just \emph{locality}. The experiments done so far give strong evidence that Bell inequalities are violated as predicted by quantum mechanics~\cite{aspect99} and thus that there can not be \emph{any} local model for the corresponding correlations: the world is in this sense nonlocal. 

Bell's theorem is a metatheoretical result of, at first sight, rather philosophical concern. That it can be of actual practical interest came as a surprise~\cite{footnote.bell}. Starting with an intuition of Ekert in 1991~\cite{ekert91}, it was established later on that nonlocal correlations can be made the basis of so-called \emph{device independent quantum information processing}~\cite{barrett05,acin07,pironio09,pironio10,colbeck11}. There are many examples in the history of physics where a theoretical result is used for applications not foreseen by the pioneers, but Bell's theorem with its foundational character might be the most striking yet. 


Whereas the violation of a Bell inequality implies nonlocality, it does not quantify it. To take the amount of the violation as a measure is problematic. E.g.,~the state which violates it maximally depends on the particular inequality~\cite{liang11}. A more natural measure of nonlocality, proposed three times independently~\cite{maudlin92,brassard99,steiner00}, quantifies it in terms of \emph{communication complexity}~\cite{footnote.commcompl}: how many bits of communication are minimally needed for the simulation of quantum states. For the simplest case of pure two-qubit states most answers are known. Surprisingly, as shown by Toner and Bacon~\cite{toner03}, a single bit is already enough for the maximally entangled state and two bits for all partially entangled ones. The main open question is whether one bit is already enough even for partial entanglement or whether the two bits of \cite{toner03} are necessary. Whatever the answer to this question, communication complexity is clearly a very coarse measure for the nonlocality of these states. In that respect a more promising approach starts with the so-called \emph{local content}.  

\begin{figure}[htbp]
\centering
\epsfig{file=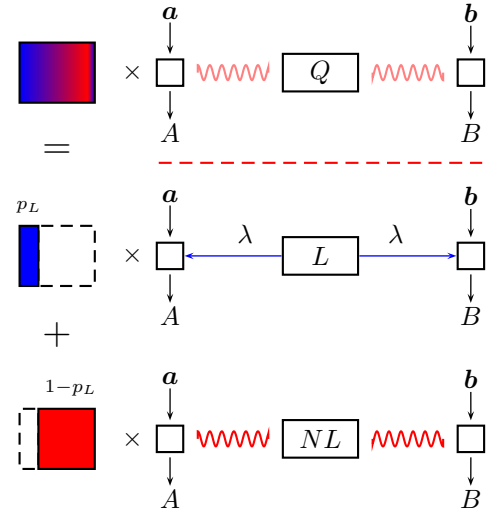}
\caption{An EPR2 decomposition holds for all measurement directions $\bm{a}$ and $\bm{b}$ and measurement outcomes $A$ and $B$. The red wavy lines indicate nonlocality. The nonlocal part is more nonlocal than the decomposed quantum probability distribution. The outcomes of the local part are determined by $\lambda$ and the corresponding measurement direction.}
\label{boxepr2}
\end{figure}

Consider an experiment where one measures the probability distribution $P_Q(A,B|\bm{a},\bm{b})$ of the measurement outcomes $A$ and $B$ for measurements $\bm{a}$ and $\bm{b}$ by preparing repeatably particle pairs in the same quantum state. If $P_Q$ comes from a nonlocal state the whole ensemble of particle pairs cannot be described with a local model. However, as first observed by Elitzur, Popescu, and Rohrlich (EPR2)~\cite{elitzur92}, there are states for which it is possible to regard a fraction $p_L \in ]0,1]$ of the particle pairs as behaving locally, and the completing fraction $1-p_L$ as behaving nonlocally. We will then say that there exists an \emph{EPR2 decomposition} into a local part $P_L$ with weight $p_L$ and a nonlocal part $P_{NL}$ with weight $1-p_{L}$ (see figure~\ref{boxepr2}):
\begin{align}\label{epr2}
P_Q=p_L P_L+(1-p_L)P_{NL}, 
\end{align} 
where the decomposition holds for all von Neumann measurements and all outcomes.
EPR2 decompositions are in general not unique. What is unique is the $\emph{local content}$ $p_{LC}$ which is defined as the maximal value $p_L$ can take over all EPR2-decompositions. We will refer to a decomposition with $p_L=p_{LC}$ as a \emph{maximal EPR2 decomposition} and define the \emph{nonlocality} of a state (or its nonlocal content) as $1-p_{LC}$. This nonlocality measure can be defined for any probability distribution. For instance, the nonlocality of the nonlocal part of a maximal EPR2 decomposition is $1$ and such a nonlocal part is thus maximally nonlocal. We also note that the nonlocal part of any EPR2 decomposition of a quantum state is always more nonlocal than or equally nonlocal as the quantum state it decomposes (see figure~\ref{boxepr2}) and is in general not quantum; but it is still non-signaling, since both $P_Q$ and $P_L$ are. 

The local content of pure two-qubit states was already investigated by EPR2. By making appropriate choices of the basis, all these states can be written as
\begin{align}\label{state}
|\Psi(\theta)\rangle=\cos\left(\frac{\theta}{2}\right)|00\rangle+\sin\left(\frac{\theta}{2}\right)|11\rangle
\end{align} 
with $\theta\in [0,\pi /2]$. All entangled states ($\theta>0$) are nonlocal~\cite{capasso73,gisin91}. In spite of this, EPR2 found in 1992 an explicit EPR2-decomposition with local weight $\frac{1}{4}(1-\sin\theta)$ for all these states, thus giving a {\it lower bound} on their local content. Furthermore, they gave an argument that the maximally entangled state has a local content of zero and is thus fully nonlocal. Further progress was achieved only in 2006 by Barrett, Kent, and Pironio~\cite{barrett06}. Since the local part necessarily obeys Bell inequalities, this places constraints on any EPR2 decomposition and can be used for deriving an {\it upper bound} on the local content. With this idea, and using a chained Bell inequality~\cite{pearle70,braunstein90}, they were able to rederive (without a technical assumption used by EPR2) the result that maximally entangled states are fully nonlocal~\cite{footnote.barrett}. Next, Scarani~\cite{scarani08} narrowed the set of possible values of $p_{LC}$ further down: He derived both a new lower bound of $1-\sin\theta$, and numerically, using again the same chained Bell inequality as~\cite{barrett06}, an upper bound of $\cos\theta$; he conjectured in the preprint version~\cite{scarani07} of~\cite{scarani08} that this equals the local content, i.e. that $p_{LC}=\cos \theta$. This conjecture was proven to be correct by Branciard, Gisin, and Scarani~\cite{branciard102} for the states with $\cos\theta \in [0,\frac{4}{5}]$. 

The main result of this article is an explicit model for the local part which generalizes the model of~\cite{branciard102} and leads to a valid EPR2 decomposition with local weight $p_L=\cos\theta$ for all states~\eqref{state}. Since $\cos\theta$ is an upper bound for the local content (which we prove analytically in Appendix~\ref{analytic}), this implies that $p_{LC}=\cos\theta$ and thus proves Scarani's conjecture for all pure two-qubit states.

\paragraph{The problem}
We will parametrize the von Neumann measurements of Alice and Bob by two three-dimensional unit vectors $\bm{a}$ and $\bm{b}$ on the Bloch sphere, in the form~\cite{footnote.orientation}
\begin{align}
\bm{a}=\begin{pmatrix} a_{\perp}\cos\alpha \\
a_{\perp}\sin\alpha \\
a_z
\end{pmatrix}, \qquad
\bm{b}=\begin{pmatrix} b_{\perp}\cos\beta \\
b_{\perp}\sin\beta \\
b_z
\end{pmatrix}.
\end{align}
With the definitions $s:=\sin \theta$, $c:=\cos\theta$, and $\chi:=\alpha-\beta$, the expectation values of Alice and Bob's measurement outcomes $A,B$ $\in\{+1,-1\}$ for the states~\eqref{state} are 
%
\begin{align}
\langle AB \rangle_{Q|\bm{a},\bm{b}}&= a_z b_z + s \, a_{\perp}b_{\perp}\cos{\chi}, \label{AB} \\
\langle A \rangle_{Q|\bm{a}}&= c \, a_z, \label{marg1}\\
\langle B \rangle_{Q|\bm{b}}&= c \, b_z.\label{marg2}
\end{align}


Without loss of generality~\cite{fine822,werner01}, we only consider deterministic models for the local part. That is, we take the outcomes to be functions of some shared random variable $\lambda$ and of the corresponding measurement directions:
\begin{align}\label{cand}
A=A(\bm{a},\lambda), \quad
B=B(\bm{b},\lambda).
\end{align}
Because of locality, $A$ is independent of $\bm{b}$, and $B$ independent of $\bm{a}$. Moreover, the probability density $\rho (\lambda)$ of the shared random variable is assumed to be independent of both measurement directions.
The expectation values of the local part are then given as
\begin{align}
\langle AB \rangle_{L|\bm{a},\bm{b}}&=\int d\lambda \, \rho(\lambda) \, A(\bm{a},\lambda) B(\bm{b},\lambda), \nonumber \\
\langle A \rangle_{L|\bm{a}}&=\int d\lambda \, \rho(\lambda) \, A(\bm{a},\lambda),\nonumber \\
\langle B \rangle_{L|\bm{b}}&=\int d\lambda \, \rho(\lambda) \, B(\bm{b},\lambda).
\end{align}

Ultimately, our aim is to prove that $p_{LC}=c$. Since $p_{LC}\leq c$, as proven numerically in \cite{scarani08} and analytically in Appendix~\ref{analytic}, it is enough to find a model for the local part that admits a local weight $p_L=c$.

Barrett, Kent, and Pironio \cite{barrett06} showed that the maximally entangled state, being fully nonlocal, necessarily has random marginals. Since the nonlocal part of a maximal decomposition of the states \eqref{state} is also fully nonlocal, its marginals have to be random, too (we prove this in Appendix~\ref{marginals}). As can be seen from \eqref{marg1} and \eqref{marg2}, this implies the constraints $\braket{A}_L=a_z$ and $\braket{B}_L=b_z$ for the local part. Note that this agrees with the ansatz used in \cite{branciard102}. 

Once one has found candidate functions~\eqref{cand} for the local part, it remains to be checked, whether the implied nonlocal part $P_{NL} = \frac{P_Q-cP_L}{1-c}$ (assuming $p_L = c$~\cite{footnote.c.in.0.1}) is a probability for all measurement directions and outcomes. Since the nonlocal part has random marginals, this is the case if the correlation term is between $-1$ and $1$.
The problem to be solved reduces to the following task~\cite{branciard102}, for any given value of $c$:
\begin{problem} Find functions $A(\bm{a},\lambda)$ and $B(\bm{b},\lambda)$ taking values in $\{+1,-1\}$, and a probability distribution $\rho(\lambda)$ such that
\begin{align}
\langle A\rangle_{L|\bm{a}}& \ = \ a_z, \label{constr.1}\\
\langle B\rangle_{L|\bm{b}}& \ = \ b_z, \label{constr.2}\\
\left| \langle AB\rangle_{Q|\bm{a},\bm{b}}-c\langle AB\rangle_{L|\bm{a},\bm{b}}\right|& \ \leq \ 1-c \label{constr.3},
\end{align}
hold for all measurement directions $\bm{a}$ and $\bm{b}$.
\end{problem}
Note that this formulation of the problem is redundant, since, as shown in Appendix~\ref{marginals}, constraint \eqref{constr.3} implies \eqref{constr.1} and \eqref{constr.2}.

\paragraph{The model.}
Our starting point is the model for the local part introduced in~\cite{branciard102}, which we will call the BGS model. The complete model below will be a generalization of this model.

\begin{modelloc} Alice and Bob share the random variable $\bm{\lambda}$ which is a random three-dimensional unit vector uniformly distributed over the Bloch sphere. For measurement directions $\bm{a}$ and $\bm{b}$, Alice and Bob output, respectively,
\begin{align}
A(\bm{a},\bm{\lambda})&=\sgn[a_z-\bm{a}\cdot\bm{\lambda}],\label{AL}\\
B(\bm{b},\bm{\lambda})&=\sgn[b_z-\bm{b}\cdot\bm{\lambda}]\label{BL},
\end{align}
where $\sgn$ is the sign function ($\in \{+1,-1\}$).
\end{modelloc}
 This model gives the right marginals (\eqref{constr.1} and \eqref{constr.2}) and the correlation becomes (see~\cite{branciard102} or Eq.~\eqref{EL_BGS} in Appendix~\ref{der.rho} for the explicit form of $E_L^{BGS}(a_z,b_z,\chi)$):
\begin{align}\label{BGS}
\braket{AB}_{L|\bm{a},\bm{b}}^{BGS}=\left\{ \begin{array}{ll}
1-|a_z-b_z|,&\ \mbox{ if }\chi=0, \\
-1+|a_z+b_z|,&\ \mbox{ if }\chi=\pi, \\
E_L^{BGS}(a_z,b_z,\chi),&\ \mbox{ if }0<|\chi|<\pi.
\end{array}\right. 
\end{align}
This model solves the problem partially: as verified numerically in~\cite{branciard102}, \eqref{constr.3} is fulfilled for all states such that $c \in [0,\frac{4}{5}]$.

To get a leading idea for our generalization of the BGS model, we will concentrate on a necessary condition for the local part~\cite{branciard102}. If Alice measures in a direction $\bm{a}$ and gets the outcome $+1$, Bob's qubit collapses to the state $\ket{\bm{a}(1)}$ (the eigenstate of $\bm{a}(1)\cdot \bm{\sigma}$ with eigenvalue $+1$), where $\ket{\bm{a}(1)}$ is one step above $\ket{\bm{a}}$ on the ``Hardy ladder''~\cite{hardy93, boschi97} (see Appendix~\ref{hardy}). The vector $\bm{a}(1)$ is given by:
\begin{align} \label{a1}
\bm{a}(1)=\begin{pmatrix}
a_{\perp}(1) \cos\alpha \\
a_{\perp}(1)\sin\alpha \\
a_z(1)
\end{pmatrix} \mbox{ with }a_z(1):=\frac{a_z+c}{1+ca_z},
\end{align}
and $a_{\perp}(1)=\sqrt{1-a_z(1)^2}$.
If Bob now measures in the direction $\bm{a}(1)$, he will always get the outcome $+1$ and thus $P_Q(+1, -1|\bm{a},\bm{a}(1))=0$.
In a similar way, we also have $P_Q(-1, +1|\bm{a},\bm{a}(-1))=0$, where
\begin{align} \label{am1}
\bm{a}(-1)=\begin{pmatrix}
a_{\perp}(-1) \cos\alpha \\
a_{\perp}(-1)\sin\alpha \\
a_z(-1)
\end{pmatrix} \mbox{ with }a_z(-1):=\frac{a_z-c}{1-ca_z},
\end{align}
and $\bm{a}(-1)$ is one step below $\bm{a}$ on the Hardy ladder. Since $P_Q=p_L P_L+(1-p_L)P_{NL}$, this implies (for $p_L\neq 0$) that the corresponding probabilities for the local part must be zero, too:
\begin{align}
P_L(+1, -1|\bm{a}, \bm{a}(1))&=0, \label{con.1}\\
P_L(-1, +1|\bm{a}, \bm{a}(-1))&=0. \label{con.2}
\end{align}
The BGS model fulfills these conditions for all states, also for the ones which violate constraint \eqref{constr.3} and thus for which the model does not give a valid EPR2 decomposition. To see this, we note for instance that $\braket{A}_{L|\bm{a}}^{BGS}= a_z$, $\braket{B}_{L|\bm{a}(1)}^{BGS}= a_z(1)$, and that, since $\bm{a}$ and $\bm{a}(1)$ have the same azimuth angle (i.e., $\chi=0$) and $a_z(1)\geq a_z$,
\begin{align}
\braket{AB}_{L|\bm{a},\bm{a}(1)}^{BGS}=1-|a_z-a_z(1)|=1+a_z-a_z(1).
\end{align} 
Hence, one indeed has
\begin{align}
&P_L^{BGS}(+1, -1|\bm{a}, \bm{a}(1))\nonumber \\
&=\frac{1}{4}\left(1+\braket{A}_{L|\bm{a}}^{BGS}-\braket{B}_{L|\bm{a}(1)}^{BGS}-\braket{AB}_{L|\bm{a},\bm{a}(1)}^{BGS}\right)=0. \
\end{align}

Let us now see in what way we can alter Alice's output function~\eqref{AL} in the BGS model, such that conditions \eqref{con.1} and \eqref{con.2} still hold (we ignore condition \eqref{constr.1} for a moment). Assume that Bob measures in direction $\bm{b}=\bm{a}(1)$: Condition (\ref{con.1}) still holds if Alice chooses to output $A(\bm{a}(t),\bm{\lambda})$ with
\begin{align}
\bm{a}(t)=\begin{pmatrix} a_{\perp}(t)\cos\alpha \\
a_{\perp}(t)\sin\alpha \\
a_z(t)
\end{pmatrix},
\end{align}
and $-1\leq a_z(t)\leq a_z(1)$. Similarly, assume that Bob measures in direction $\bm{b}=\bm{a}(-1)$: Condition (\ref{con.2}) still holds if Alice chooses to output $A(\bm{a}(t),\bm{\lambda})$ with $a_z(-1)\leq a_z(t)\leq 1$. Taken together, since Bob could be measuring in either direction $\bm{a}(1)$ or $\bm{a}(-1)$ about which Alice has no information, Alice can choose to output $A(\bm{a}(t),\bm{\lambda})$ with $a_z(-1)\leq a_z(t)\leq a_z(1)$ and still be sure that conditions (\ref{con.1}) and (\ref{con.2}) hold. We will say that Alice is free to choose $\bm{a}(t)$ in the {\it Hardy sector} $\mathcal{H}_{\bm{a}}$ of $\bm{a}$.
To use this wiggle room is precisely the idea for the construction of our $\theta$-dependent generalization of the BGS model. 

Bob's output will be identical to his output \eqref{BL} but Alice will now output $A(\bm{a}(t),\bm{\lambda})$, where $t\in [-1,1]$ parametrizes the vectors in the Hardy sector $\mathcal{H}_{\bm{a}}$, such that $\bm{a}(0)=\bm{a}$, and that $\bm{a}(1)$ and $\bm{a}(-1)$ are the two vectors introduced in~(\ref{a1}--\ref{am1}). Alice picks a particular $\bm{a}(t)$ depending on a local (nonshared) random variable $t$, which has a probability distribution $\rho_{\bm a}(t)$ that depends on $a_z$ (besides depending on $\theta$)~\cite{footnote.colbeck}. One obvious constraint on $\rho_{\bm a}(t)$ comes from the condition that \eqref{constr.1} should still be fulfilled. Note that for $\chi=0$, the model implies the same correlation term \eqref{BGS} as the BGS model if $\bm{b}$ does \emph{not} lie in the Hardy sector $\mathcal{H}_{\bm{a}}$ of $\bm{a}$. Incidentally, for such measurement directions, one can show (see Appendix~\ref{app.BGS}) that the BGS model does lead to a valid EPR2 decomposition for all values of $\theta$, as it should be for the following generalization having a chance to be valid for all states and all measurement directions.

\begin{modellocgen} Alice and Bob share the random three-dimensional vector $\bm{\lambda}$, uniformly distributed on the Bloch sphere. Alice can sample an additional local random variable $t \in [-1,1]$ according to the following distribution, which depends on her setting $\bm{a}$ through its component $a_z$:
\begin{align}
&\rho_{\bm a}(t):= \, \frac{s\ln(\gamma)}{4c} \, \frac{1+G(t) a_z}{\sqrt{1-G(t)^2}} \, , \\[1mm]
&\mbox{ with } \quad \gamma:=\frac{1+c}{1-c}, \quad G(t):=\frac{\gamma^t-1}{\gamma^t+1}.
\end{align}

Alice and Bob's outputs are then defined as
\begin{align}
A(\bm{a},\bm{\lambda},t)&=\sgn[a_z(t)-\bm{a}(t)\cdot\bm{\lambda}],\\
B(\bm{b},\bm{\lambda})&=\sgn[b_z-\bm{b}\cdot\bm{\lambda}],
\end{align}
where
\begin{align}
\bm{a}(t):=\begin{pmatrix} a_{\perp}(t)\cos\alpha \\
a_{\perp}(t)\sin\alpha \\
a_z(t)
\end{pmatrix} \mbox{ with } \ a_z(t):=\frac{a_z+G(t)}{1+G(t)a_z},
\end{align}
and $a_{\perp}(t) = \sqrt{1-a_z(t)^2}$.
\end{modellocgen}

The parametrization of the vectors in the Hardy sectors $\mathcal{H}_{\bm{a}}$ as $\bm{a}(t)$ is motivated in Appendix~\ref{hardy}. The probability distribution $\rho_{\bm a} (t)$ is derived in Appendix~\ref{der.rho} from the postulate that constraint \eqref{constr.3} is saturated for the case that $\bm{b}$ lies in the Hardy sector $\mathcal{H}_{\bm{a}}$ of $\bm{a}$. This, together with the previous remark on the case where $\bm{b}$ does not lie in $\mathcal{H}_{\bm{a}}$, allows us to prove analytically that our local model above defines a valid EPR2 decomposition, with a maximal local weight $p_L = p_{LC} = c$, for all measurement directions $\bm{a},\bm{b}$ with $\chi = 0$. One can show in a similar way that this also works for measurements with $\chi = \pi$.

For measurement settings $\bm{a},\bm{b}$ with $0 < |\chi| < \pi$, we checked numerically with standard numerical optimization tools of a common computational software program that constraint \eqref{constr.3} holds within the expected numerical precision (of absolute order $10^{-10}$). Note that for any $\theta$, $\braket{AB}_{L|\bm{a},\bm{b}}$ depends only on the three parameters $a_z$, $b_z$, and $\chi$.
We also note that the implied nonlocal part is nonquantum for all $c \in \, ]0,1[$ (see Appendix~\ref{nonquantum}).

\paragraph{Discussion.} Our new EPR2 decomposition for two-qubit states, inspired by that of~\cite{branciard102} and based on properties of Hardy ladder vectors, thus allows us to show that the local content of all pure two-qubit states~\eqref{state} equals $c = \cos\theta$. This proves a conjecture formulated in~\cite{scarani07}. Interestingly, the nonlocality $1-c$ of these states is found to be monotonously related to its entanglement --- against the trend of results which suggest that entanglement and nonlocality are of a rather different nature (see e.g.~\cite{brunner05,methot07,junge11}). 

One obvious next step is to consider more general states. For the case of higher dimensional partially entangled two-partite states, we only know of the numerical lower bound on $p_{LC}$ found by Scarani~\cite{scarani08} for a class of entangled qutrits. We believe that the methods presented in appendices \ref{analytic} and \ref{marginals} should be generalizable to get analytical upper bounds which might then pave the way for models of the local part. Whether the monotonous connection between entanglement and nonlocality then still holds is a question of particular interest. 


We would like to end with one possible application of our decomposition and some ensuing open questions. \emph{Any} EPR2 decomposition can be used as a starting point for a simulation protocol, where only the nonlocal part remains to be simulated~\cite{brunner08}. Now, the arguably biggest difficulty in the simulation of partially entangled states is their nonrandom marginals. Using our maximal EPR2 decomposition seems in that respect ideally suited, since its nonlocal part has random marginals. We also note that nonlocal correlations with random marginals that can be simulated with one bit of classical communication can be simulated with one nonlocal Popescu-Rohrlich (PR) box~\cite{popescu94}, too~\cite{barrett052}. Thus, \emph{if} the nonlocal part of a maximal EPR2 decomposition for two-qubit correlations can be simulated with one bit of communication, then the corresponding quantum state can be simulated with one PR box. This is however known to be impossible for weakly entangled 2-qubit states~\cite{brunner05}; hence their nonlocal part in a maximal EPR2 decomposition cannot be simulated with one bit of communication. Does this answer the open question, whether partially entangled states can be simulate with a single bit? Not yet: Whereas it is always possible to simulate the decomposed state by simulating the nonlocal part, this might not be the most optimal way. 

On the one hand, maximal EPR2 decompositions lead to a natural alternative measure of nonlocality. On the other hand, they thus lead to new insights on the quantification of nonlocality in terms of communication complexity.




\paragraph{Acknowledgments.} We thank Jean-Daniel Bancal, Tomer Barnea, Nicolas Brunner, Rita Hidalgo Staub, Yeong-Cherng Liang, Charles Ci Wen Lim, Fabio Molo, Carlos Palazuelos, Valerio Scarani, Tim R\"az, and Denis Rosset for discussions. This work was supported by the Swiss NCCR-QSIT, the European ERC-AG QORE, and by a UQ Postdoctoral Research Fellowship.

\bibliography{epr2}

\begin{thebibliography}{10}
\bibitem{footnote.einstein}%
  \BibitemOpen
  \bibinfo {note} {Whether or not Einstein thought of a completion of quantum
  theory in the form of a hidden variable theory ruled out by Bell's theorem,
  as Bell believed, has been disputed (see the appendix of \cite{bell762} and
  e.g.~\cite{stachel86}).}%
  \bibAnnoteFile{Stop}{footnote.einstein}%
\bibitem{footnote.bell}%
  \BibitemOpen
  \bibinfo {note} {It was a surprise for John Bell, too (see
  e.g.~\cite{gilder09}, p.~314).}%
  \bibAnnoteFile{Stop}{footnote.bell}%
\bibitem{footnote.barrett}%
  \BibitemOpen
  \bibinfo {note} {It was also shown in \cite{barrett06} that this result is
  true for maximally entangled two-partite states of any dimension.}%
  \bibAnnoteFile{Stop}{footnote.barrett}%
\bibitem{footnote.orientation}%
  \BibitemOpen
  \bibinfo {note} {A technical note: in order for the expectation values
  $\langle AB \rangle_{Q|\bm{a},\bm{b}}$ to take the form of~\eqref{AB}, the
  $y$ directions of Alice and Bob's Bloch spheres should be defined with
  opposite directions (in other words: with standard definitions, one needs for
  instance to change, throughout the paper, $\bm{b}$ to $\bm{b'}$, its
  reflection with respect to the $xz$ plane; see e.g.~\cite{branciard102}).}%
  \bibAnnoteFile{Stop}{footnote.orientation}%
\bibitem{footnote.c.in.0.1}%
  \BibitemOpen
  \bibinfo {note} {For technical reasons we assume from now on that $c<1$; note
  that for $c=1$, the state~\eqref{state} is a product state, and is therefore
  fully local: $p_{LC}=1=c$.}%
  \bibAnnoteFile{Stop}{footnote.c.in.0.1}%
\bibitem{footnote.colbeck}%
  \BibitemOpen
  \bibinfo {note} {Colbeck and Renner~\cite{colbeck08} investigated models in
  which at least one party uses such a local nonshared random variable. The
  methods used in~\cite{colbeck08} can be adapted to show that the nonlocal
  part of the following maximal decomposition cannot have any such local
  nonshared random variables.}%
  \bibAnnoteFile{Stop}{footnote.colbeck}%
\end{thebibliography}

\appendix

\section{Hardy ladder vectors}\label{hardy}
The vectors $\bm{a}(1)$ and $\bm{a}(-1)$, one step above and below $\bm{a} = (a_{\perp}\cos\alpha,a_{\perp}\sin\alpha,a_z)$ on the Hardy ladder, are defined such that
$P_Q(+1,-1|\bm{a},\bm{a}(1))=0$ and $P_Q(-1,+1|\bm{a},\bm{a}(-1))=0$ respectively, and given by
\begin{align}
\bm{a}(1)=\begin{pmatrix}
a_{\perp}(1)\cos\alpha \\
a_{\perp}(1)\sin\alpha \\
a_z(1)
\end{pmatrix} \quad \mbox{with } a_z(1)=\frac{a_z+c}{ca_z+1},
\end{align}
and 
\begin{align}
\bm{a}(-1)=\begin{pmatrix}
a_{\perp}(-1)\cos\alpha \\
a_{\perp}(-1)\sin\alpha \\
a_z(-1)
\end{pmatrix} \quad \mbox{with } a_z(-1)=\frac{a_z-c}{-ca_z+1}.
\end{align} 
The successive steps of the Hardy ladder are iteratively defined as
\begin{align}
P_Q(+1,-1|\bm{a}(n),\bm{a}(n+1))&=0,\quad \mbox{and} \\
P_Q(-1,+1|\bm{a}(-n),\bm{a}(-n-1))&=0
\end{align}
for all $n\in \mathbb{N}$. To get the explicit form of $\bm{a}(n)$ for all $n\in \mathbb{Z}$, we note that the maps 
\begin{align}
a_z  \mapsto \frac{a_z\pm c}{\pm ca_z+1}
\end{align}
are (real) M\"obius transformations. The composition of M\"obius transformations basically reduces to the product of their associated matrices, in our case to the product of
\begin{align}\label{mob}
\left( \begin{array}{cc}
1& \pm c \\
\pm c & 1
\end{array}\right).
\end{align}

Defining (we assume $c<1$)
\begin{align}
\gamma:=\frac{1+c}{1-c} \geq 1, \quad G(n):=\frac{\gamma^n-1}{\gamma^n+1} \in [-1,1] ,
\end{align}
one gets the explicit form of $\bm{a}(n)$ as
\begin{align} \label{hardyvec}
\bm{a}(n)=\begin{pmatrix}
a_{\perp}(n) \cos\alpha \\
a_{\perp}(n)\sin\alpha \\
a_z(n)
\end{pmatrix} \mbox{ with } a_z(n):=\frac{a_z+G(n)}{G(n)a_z+1},
\end{align}
and $a_{\perp}(n) = \sqrt{1-a_z(n)^2}$.
This can independently be checked e.g.~by induction.

\medskip

In order to get a parametrization of the vectors ``between'' $\bm{a}(-1)$ and $\bm{a}(1)$, we let $n$ take real values $t$ in the interval $[-1,1]$. Note that $a_z(t)$ is a monotonically increasing function. We then define for every vector $\bm{a}$ its Hardy sector $\mathcal{H}_{\bm{a}}$ to be the set of all vectors which can be written as $\bm{a}(t)$ for some $t\in[-1,1]$:
\begin{align}
\mathcal{H}_{\bm{a}}:=\{\bm{a}(t)|t\in [-1,1]\}.
\end{align}




\section{Analytic proof that $c$ is an upper bound for $p_{LC}$ }\label{analytic}

Let us label $N$ measurement directions of Alice as $\bm{a}(-N)$, $\bm{a}(-N+2)$, $\dots$, $\bm{a}(N-2)$ and $N$ directions of Bob as $\bm{b}(-N+1)$, $\bm{b}(-N+3)$, $\dots$, $\bm{b}(N-1)$. For a given probability distribution $P$ with binary outcomes $A,B$, we define as a correlation measure the expression
\begin{align}\label{d.1}
I_N(P)=&P(A_{-N}\neq B_{-N+1})+P(A_{-N+2}\neq B_{-N+1}) +\dots \nonumber \\
&+P(A_{N-2}\neq B_{N-1})+P(A_{-N}=B_{N-1}),
\end{align}
where $P(A_{i}\neq B_{j})$ is a shortcut notation for $P(A\neq B|\bm{a}(i),\bm{b}(j))$.
As shown in \cite{pearle70, braunstein90}, $I_N$ is a chained Bell polynomial, defining the following chained Bell inequality---satisfied by all local correlations:
\begin{align}
I_N(P_L)\geq 1.
\end{align}
Since by linearity,
\begin{align}
I_N(P_Q)=p_L I_N(P_L)+(1-p_L)I_N(P_{NL})
\end{align}
and since $I_N(P_{NL})$ does not have any constraints except that the probabilities are between zero and one and thus $I_N(P_{NL})\geq 0$, $I_N(P_Q)$ is an upper bound for $p_L$~\cite{barrett06}:
\begin{align}\label{bound}
p_L\leq I_N(P_Q).
\end{align}

We show now that the measurement directions of Alice and Bob can be chosen such that
\begin{align}
\lim_{N\rightarrow\infty}I_N(P_Q)=c.
\end{align}
For that, we start with an arbitrary unit vector
\begin{align}\label{d.6}
\bm{v}(0) :=\begin{pmatrix}
v_{\perp}(0) \cos(\omega ) \\
v_{\perp}(0) \sin(\omega )\\
v_z(0)
\end{pmatrix},
\end{align}
for any $v_z(0) \in ]-1,1[$, $v_{\perp}(0) = \sqrt{1-v_z(0)^2}$, and $\omega \in [0,2\pi [$. 
Let us then define, with the notation of~\eqref{hardyvec},
\begin{align}
\bm{v}(n):&=\begin{pmatrix}
v_{\perp}(n) \cos(\omega ) \\
v_{\perp}(n) \sin(\omega )\\
v_z(n)
\end{pmatrix} .
\end{align}
We then choose Alice's settings to be $\bm{a}(n) = \bm{v}(n)$ for $n\in \{-N,-N+2,\dots,N-2\}$, and Bob's settings to be $\bm{b}(n) = \bm{v}(n)$ for $n\in \{-N+1,-N+3,\dots,N-1\}$. An example of the $z$ coordinates of such a set of vectors is drawn in figure \ref{hardy.ladder}, for $N=3$.
Note that we have
\begin{align}\label{lim}
\lim_{n\rightarrow \infty} \bm{v}(n)=\begin{pmatrix}
0\\0\\1 \end{pmatrix} \mbox{ and } 
\lim_{n\rightarrow -\infty} \bm{v}(n)=\begin{pmatrix}
0\\0\\-1 \end{pmatrix}.
\end{align}

\begin{figure}[htbp]
\centering
\epsfig{file=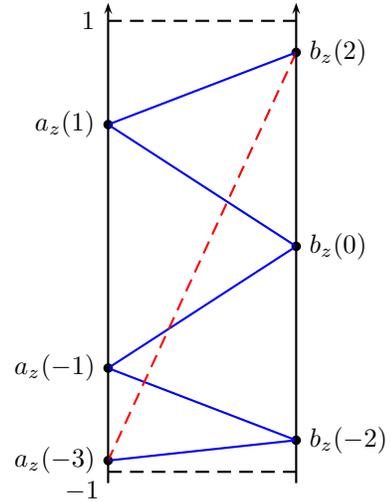}
\caption{Hardy ladder for $N=3$. Every colored line represents one term in the chained Bell inequality, the dashed red one being the last one, $P(A_{-N}=B_{N-1})$.}
\label{hardy.ladder}
\end{figure}

Since 
\begin{align}
P_Q(1, -1|\bm{a}(n), \bm{b}(n+1))=0
\end{align}
for two neighboring Hardy steps $\bm{a}(n) = \bm{v}(n)$ and $\bm{b}(n+1) = \bm{v}(n+1)$, one has
\begin{align}
\langle A_{n}B_{n+1}\rangle_Q := & \ \langle AB \rangle_{Q|\bm{a}(n), \bm{b}(n+1)} \nonumber \\
= & \ 1+c [v_z(n)-v_z(n+1)],
\end{align}
and hence
\begin{align}
P_Q(A_n\neq B_{n+1})&=\frac{1}{2}[1-\langle A_{n}B_{n+1}\rangle_Q] \nonumber \\
&=\frac{c}{2}[v_z(n+1)-v_z(n)].
\end{align}
Similarly,
\begin{align}
P_Q(A_{n+1}\neq B_n)=\frac{c}{2}[v_z(n+1)-v_z(n)].
\end{align}
$I_N(P_Q)$ then becomes
\begin{align}
I&_N(P_Q) \nonumber \\
&=\sum_{n=-N}^{N-2}\frac{c}{2}[v_z(n+1)-v_z(n)]+\frac{1}{2}\left[ 1+\langle A_{-N}B_{N-1}\rangle_Q \right] \nonumber \\
&=\frac{c}{2}[v_z(N-1)-v_z(-N)]+\frac{1}{2}\left[ 1+\langle A_{-N}B_{N-1}\rangle_Q \right].
\end{align}
And because of (\ref{lim}), we have
\begin{align}
\lim_{N\rightarrow \infty}I_N(P_Q)=c,
\end{align}
which leads with \eqref{bound} to
\begin{align}
p_L\leq c.
\end{align}

\section{Proof that $\braket{A}_{NL}=\braket{B}_{NL}=0$ for a maximal EPR2 decomposition of $P_Q$}\label{marginals}
As shown in Appendix~\ref{analytic}, there exist measurement directions for which
\begin{align}
\lim_{N\rightarrow \infty}I_N(P_Q)=c.
\end{align}
Since we have, for a maximal EPR2 decomposition of $P_Q$ (i.e., for $p_{L} = p_{LC} = c$),
\begin{align}
I_N(P_Q)=c I_N(P_L)+(1-c)I_N(P_{NL}),
\end{align}
with $I_N(P_L)\geq 1$ and $I_N(P_{NL})\geq 0$, then we must have
\begin{align}
\lim_{N\rightarrow \infty}I_N(P_L)&=1 \ \mbox{ and } \
\lim_{N\rightarrow \infty}I_N(P_{NL})=0
\end{align}
for such a set of measurement directions.

Note that apart from the last term $P(A_{-N}=B_{N-1})$, a term that is present in $I_N(P_{NL})$ is also present in all subsequent polynomials $I_{N+2k}(P_{NL})$, for $k \in {\mathbb{N}}$.
Since all terms appearing in $I_N(P_{NL})$ are probabilities, and since $\lim_{N\rightarrow\infty}I_N(P_{NL}) = 0$, then all terms have to be zero, except possibly for the last one which must only {\it tend to} zero (as $N\rightarrow\infty$).

Take one term: $P(A_{n}\neq B_{n\pm 1}) = 0$ implies that
\begin{align}
& P(A_{n}=+1,B_{n\pm 1}=-1) \nonumber \\
& \quad = \frac{1}{4}(1+\langle A_{n}\rangle_{NL}-\langle B_{n\pm 1}\rangle_{NL}-\langle A_{n}B_{n\pm 1}\rangle_{NL}) = 0 , \nonumber \\
& P(A_{n}=-1,B_{n\pm 1}=+1) \nonumber \\
& \quad = \frac{1}{4}(1-\langle A_{n}\rangle_{NL}+\langle B_{n\pm 1}\rangle_{NL}-\langle A_{n}B_{n\pm 1}\rangle_{NL}) = 0, \nonumber
\end{align}
and therefore
\begin{align}
& P(A_{n}=+1,B_{n\pm 1}=-1) - P(A_{n}=-1,B_{n\pm 1}=+1) \nonumber \\
& \qquad \qquad = \frac{1}{2}(\langle A_{n}\rangle_{NL}-\langle B_{n\pm 1}\rangle_{NL}) = 0.
\end{align}
Hence, $\langle A_{n}\rangle_{NL} = \langle B_{n\pm 1}\rangle_{NL}$, and by induction $\langle A_{i}\rangle_{NL} = \langle B_{j}\rangle_{NL}$ for any odd value of  $i$ and even value of $j$.

In a similar way, $P(A_{-N}=B_{N-1}) \rightarrow 0$ implies that
\begin{align}
& P(A_{-N}=B_{N-1}=+1) \nonumber \\
& = \frac{1}{4}(1+\langle A_{-N}\rangle_{NL}+\langle B_{N-1}\rangle_{NL}+\langle A_{-N}B_{N-1}\rangle_{NL}) \rightarrow 0 , \nonumber \\
& P(A_{-N}=B_{N-1}=-1) \nonumber \\
& = \frac{1}{4}(1-\langle A_{-N}\rangle_{NL}-\langle B_{N-1}\rangle_{NL}+\langle A_{-N}B_{N-1}\rangle_{NL}) \rightarrow 0, \nonumber
\end{align}
and therefore
\begin{align}
& P(A_{-N}=B_{N-1}=+1) - P(A_{-N}=B_{N-1}=-1) \nonumber \\
& \qquad \qquad = \frac{1}{2}(\langle A_{-N}\rangle_{NL}+\langle B_{N-1}\rangle_{NL}) \rightarrow 0. \label{limit_marg}
\end{align}

Now, from the previous argument, $\langle A_{-N}\rangle_{NL} = \langle B_{N-1}\rangle_{NL} = \langle A_{i}\rangle_{NL} = \langle B_{j}\rangle_{NL}$ for any $i,j$. Eq.~\eqref{limit_marg} then implies that $\langle A_{i}\rangle_{NL} = \langle B_{j}\rangle_{NL} = 0$ for all the measurement directions appearing in $I_{N}(NL)$.

Since we started the construction in Appendix~\ref{analytic} with an arbitrary measurement vector (see equation~\eqref{d.6}), we must have, for the maximal EPR2 decomposition considered here,
\begin{align}
\braket{A}_{NL|\bm{a}}=\braket{B}_{NL|\bm{b}}=0
\end{align}
for all $\bm{a}$ and $\bm{b}$. This in turn implies that the marginals of the local part are, as given in Eq.~(\ref{constr.1}--\ref{constr.2}),
\begin{align}
\langle A\rangle_{L|\bm{a}} \ = \ a_z, \quad \langle B\rangle_{L|\bm{b}} \ = \ b_z.
\end{align}

\section{Proof that the BGS model gives a valid EPR2 decomposition for measurement settings such that $\chi=0$, and $b \not\in \mathcal{H}_{a}$}\label{app.BGS}

In this Appendix we prove that when Alice and Bob's settings are such that $\chi=0$, but such that $\bm{b}$ does not lie in the Hardy sector $\mathcal{H}_{\bm{a}}$ of $\bm{a}$ (i.e., $b_z > a_z(1)$ or $b_z < a_z(-1)$), then the BGS EPR2 decomposition~\cite{branciard102} is valid for all states. For that, we will show that~\eqref{constr.3} holds for all values of $c$.

Suppose indeed that $\chi=0$ and $b_z > a_z(1) = \frac{a_z+c}{1+c a_z} \geq a_z$. Then, omitting the conditional settings, we have
\begin{align}
\langle AB\rangle_{Q} - c \langle AB\rangle_{L} = a_z b_z + s a_\perp b_\perp - c(1 - b_z + a_z). \nonumber
\end{align}
Remembering that $b_\perp = \sqrt{1-b_z^2}$, we find
$\frac{\partial^2}{\partial b_z^2}[\langle AB\rangle_{Q} - c \langle AB\rangle_{L}] = -s a_\perp b_\perp^{-3} \leq 0$, 
and therefore that $\frac{\partial}{\partial b_z}[\langle AB\rangle_{Q} - c \langle AB\rangle_{L}]$ decreases with $b_z$. Now, observing that $\frac{\partial}{\partial b_z}[\langle AB\rangle_{Q} - c \langle AB\rangle_{L}]_{|b_z = a_z(1)} = 0$, we conclude that for $b_z \geq a_z(1)$, $\frac{\partial}{\partial b_z}[\langle AB\rangle_{Q} - c \langle AB\rangle_{L}] \leq 0$. Hence, $\langle AB\rangle_{Q} - c \langle AB\rangle_{L}$ decreases monotonically for $b_z \in [a_z(1),1]$, and therefore takes values between $[\langle AB\rangle_{Q} - c \langle AB\rangle_{L}]_{|b_z = 1} = (1-c)a_z \geq -(1-c)$ and $[\langle AB\rangle_{Q} - c \langle AB\rangle_{L}]_{|b_z = a_z(1)} = 1-c$. It follows that $|\langle AB\rangle_{Q} - c \langle AB\rangle_{L}| \leq 1-c$, i.e. that~\eqref{constr.3} is satisfied.

The case $b_z < a_z(-1)$ can be analyzed in a very similar way. We thus conclude that when Alice and Bob's settings are such that $\chi=0$, but with $\bm{b} \notin \mathcal{H}_{\bm{a}}$, the BGS model gives a valid EPR2 decomposition for all states.

\section{Derivation of $\rho_{a}(t)$}\label{der.rho}

We consider the case $\chi=0$, and assume for now that $\bm{b}$ lies in the Hardy sector $\mathcal{H}_{\bm{a}}$ of $\bm{a}$.  
The correlator reads 
\begin{align}
\langle AB\rangle_{L|\bm{a},\bm{b}}^{\chi=0}=\int\limits_{-1}^1 dt \rho_{\bm a}(t) (1-|a_z(t)-b_z|).
\end{align}
To determine $\rho_{\bm a}(t)$ we will demand that condition \eqref{constr.3} is exactly saturated.

\begin{prop}
If $\rho_{\bm a}(t):= \frac{s\ln(\gamma)}{4c} \frac{1+G(t) a_z}{\sqrt{1-G(t)^2}}$, then $\langle AB\rangle_{Q|\bm{a},\bm{b}}^{\chi=0}-c\langle AB\rangle_{L|\bm{a},\bm{b}}^{\chi=0}=1-c$.
\end{prop}
\begin{proof}
We will use the ansatz, for $t \in [-1,1]$,
\begin{align}\label{ansatz}
\rho_{\bm a}(t)= [1+G(t)a_z] \, \rho(t)
\end{align}
where $\rho(t)$ is a (normalized) even probability distribution, and we further assume that
\begin{align}
\int\limits_{0}^{\pm 1} dt \rho(t) G(t)&=\frac{1-s}{2c} \ \ \Big( =\frac{1}{2}G(\frac{1}{2}) \Big) \label{cond.1}.
\end{align}
The ansatz~\eqref{ansatz} ensures that $\rho_{\bm a}(t)$ is a well-defined probability distribution and that $\braket{A}_L=a_z$ as desired (which follow from the facts that $1+G(t)a_z \geq 0$ and that $\rho(t)G(t)$ is an odd function), while~\eqref{cond.1} ensures that the proposition is true for $\bm{a}=\bm{b}$. At the end, we will have to check that condition (\ref{cond.1}) holds and that $\rho(t)$ is indeed a probability distribution.

Since $\bm{b}$ is assumed to lie in $\mathcal{H}_{\bm{a}}$, $b_z$ can then be written as 
\begin{align}
b_z=a_z(T) \mbox{ for some }T\in [-1,1].
\end{align}
We have
\begin{align}
\langle AB\rangle_{L|\bm{a},\bm{b}}^{\chi=0}=& 1 - \int\limits_{-1}^1 dt \rho_{\bm a}(t) \left| a_z(t)-b_z\right| \\
=& 1 + \int_{-1}^T dt f(t) - \int_{T}^1 dt f(t)
\end{align}
with $f(t) = \rho_{\bm a}(t) \left[ a_z(t)-b_z\right] = \rho(t) \left[ a_z-b_z + (1-a_z b_z) G(t) \right]$. 
By decomposing the two integrals, we find
\begin{align}
\langle AB\rangle_{L|\bm{a},\bm{b}}^{\chi=0}=& 1 + \int_{-1}^0 dt f(t) + 2\int_{0}^T dt f(t) - \int_{0}^1 dt f(t) \nonumber \\
=& 1 - \int_{0}^1 dt [f(t) - f(-t)] + 2\int_{0}^T dt f(t) \nonumber \\
=& 1 - (1-a_z b_z)\frac{1-s}{c} + 2\int_{0}^T dt f(t),
\end{align}
where we made use of the parity of $\rho(t)$ and of assumption~\eqref{cond.1}.

Demanding that $\langle AB\rangle_Q^{\chi=0}-c\langle AB\rangle_L^{\chi=0}=1-c$ with
\begin{align}
\langle AB\rangle_Q^{\chi=0}=a_z b_z + s \, a_{\perp}b_{\perp}
\end{align}
leads to
\begin{align}\label{inteq}
2c \int\limits_{0}^T dt f(t) =s a_{\perp}b_{\perp}-s (1-a_z b_z).
\end{align}
Expressing now $b_z$ as $\frac{a_z+G(T)}{1+G(T) a_z}$, one gets
\begin{align}
a_{\perp}b_{\perp}&=\frac{a_{\perp}^2}{1+G(T)a_z}\sqrt{1-G(T)^2}, \nonumber  \\
a_z-b_z&=\frac{-a_{\perp}^2}{1+G(T)a_z}G(T), \nonumber \\
1-a_zb_z&=\frac{a_{\perp}^2}{1+G(T)a_z}.
\end{align}
Using these expressions, \eqref{inteq} leads to the integral equation
\begin{align}
\int_0^T dt \rho(t) [G(T)-G(t)]=\frac{s}{2c}\left(1-\sqrt{1-G(T)^2}\right),
\end{align}
which can be solved easily: e.g., after substituting 
\begin{align}
t\rightarrow G(t),
\end{align}
one gets an equation with a linear kernel. The unique solution reads:
\begin{align}
\rho(t)=\frac{s\ln(\gamma)}{4c} \frac{1}{\sqrt{1-G(t)^2}}.
\end{align}
As can be checked, $\rho(t)$ is an even, normalized probability distribution on the interval $[-1,1]$, and condition (\ref{cond.1}) holds.

\end{proof}

The choice $\rho_{\bm a}(t):= \frac{s\ln(\gamma)}{4c} \frac{1+G(t) a_z}{\sqrt{1-G(t)^2}}$ thus ensures that our maximal EPR2 decomposition, with $p_L = p_{LC} = c$, is valid for all measurement directions $\bm{a}$ and $\bm{b}$ such that $\chi=0$, when $\bm{b}$ lies in the Hardy sector $\mathcal{H}_{\bm{a}}$ of $\bm{a}$. If $\chi=0$ but $\bm{b} \notin \mathcal{H}_{\bm{a}}$, our local model gives the same correlation as the BGS model, which we proved in Appendix~\ref{app.BGS} to also give a valid EPR2 decomposition. Note that if $\chi = \pi$, the situation is quite similar to $\chi = 0$, and one can prove again that our EPR2 decomposition is still valid.

The question now is whether our maximal EPR2 decomposition remains valid for all other settings, when $0<|\chi|<\pi$. The difficulty in checking that~\eqref{constr.3} indeed holds is that one needs to calculate $\langle AB\rangle_{L|\bm{a},\bm{b}}$, given by
\begin{align}
\langle AB\rangle_{L|\bm{a},\bm{b}} = \int_{-1}^1 dt \rho_{\bm a}(t) \, E_L^{BGS}(a_z(t),b_z,\chi), \label{EL_complete}
\end{align}
with
\begin{align}
E_L^{BGS}(a_z,&b_z,\chi) \nonumber \\
=&1-\frac{2|\chi|}{\pi}+\frac{2}{\pi}a_z \arctan \left(\frac{a_{\perp}b_z-a_zb_{\perp}\cos\chi}{b_{\perp}\sin|\chi |}\right)\nonumber \\
&+\frac{2}{\pi}b_z \arctan \left( \frac{a_zb_{\perp}-a_{\perp}b_z\cos\chi}{a_{\perp}\sin|\chi|}\right) \label{EL_BGS}
\end{align}
from the BGS model~\cite{branciard102}.

Unfortunately we were not able to calculate the integral~\eqref{EL_complete} explicitly. However, we carried intensive numerical checks to convince ourselves that~\eqref{constr.3} always holds (up to machine precision of absolute order $10^{-10}$), and therefore that our maximal EPR2 decomposition is valid for all possible measurement settings.

\section{Maximal violation of the CHSH inequality of the nonlocal part}\label{nonquantum}
The correlator of the nonlocal part $P_{NL}$ of our EPR2 decomposition is given by:
\begin{align}
E_{NL}(\bm{a},\bm{b}):=\braket{AB}_{NL|\bm{a},\bm{b}} =\frac{\braket{AB}_{Q|\bm{a},\bm{b}}-c\braket{AB}_{L|\bm{a},\bm{b}}}{1-c} . \nonumber
\end{align}
\begin{figure}[htbp]
\centering
\epsfig{file=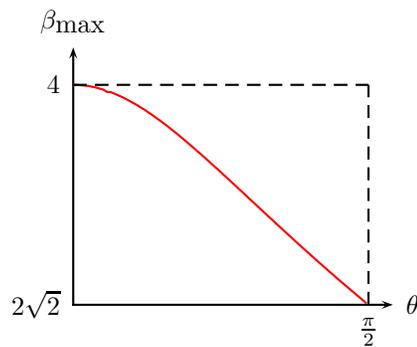}
\caption{Maximum violation $\beta_{\mbox{max}}$ of the CHSH inequality by the nonlocal part $P_{NL}$ of our EPR2 decomposition.}
\label{fig.nonquantum}
\end{figure}
We denote with $\beta_{\mbox{max}}$ the maximum violation of the CHSH inequality~\cite{clauser69}:
\begin{align}
\beta_{\mbox{max}} \, := \, \max_{\stackrel{\bm{a}_1,\bm{a}_2}{\bm{b}_1,\bm{b}_2}} \, \Big| & E_{NL}(\bm{a}_1,\bm{b}_1)+E_{NL}(\bm{a}_1,\bm{b}_2) \nonumber \\[-5mm]
& +E_{NL}(\bm{a}_2,\bm{b}_1)-E_{NL}(\bm{a}_2,\bm{b}_2)\Big| .
\end{align}
In figure \ref{fig.nonquantum}, the numerically determined values of $\beta_{\mbox{max}}$ are plotted for all states $\theta\in [0,\pi/2]$. All states with $\theta <\frac{\pi}{2}$ violate Tsirelson's bound $2\sqrt{2}$~\cite{tsirelson80}, approaching the algebraic maximum of 4 for $\theta\rightarrow 0$; this proves that the corresponding nonlocal parts $P_{NL}$ are nonquantum.

\end{document}